\documentclass[conference,a4paper]{IEEEtran}

\usepackage{amssymb}
\usepackage{amsthm}

\usepackage[T1]{fontenc}
\usepackage{url}
\usepackage{ifthen}
\usepackage{cite}
\usepackage[cmex10]{amsmath} 

\usepackage{graphicx}
\newtheorem{theorem}{Theorem}

\newtheorem{remark}{Remark}

\newtheorem{corollary}{Corollary}
\usepackage{mathtools}

\usepackage[section]{placeins}

	\usepackage{cite}
	\usepackage[ruled,vlined,linesnumbered]{algorithm2e}

\usepackage{tikz}

\usepackage{verbatim}
\usetikzlibrary{arrows}

\newcommand{\Expect}{{\rm I\kern-.3em E}}

\newtheorem{example}{{\em Example}}

\newcommand\blfootnote[1]{%
  \begingroup
  \renewcommand\thefootnote{}\footnote{#1}%
  \addtocounter{footnote}{-1}%
  \endgroup
}
\begin{document}
\title{Receiver Caching with Cooperative Transmission in Linear Interference Networks} 
\title{Wyner's Network on Caches: Combining\\ Receiver Caching with a Flexible Backhaul}
\author{%
  \IEEEauthorblockN{Eleftherios Lampiris}
  \IEEEauthorblockA{EURECOM\\
                    Sophia Antipolis, France\\
                    Email: lampiris@eurecom.fr}
\and
  \IEEEauthorblockN{Aly El Gamal}
  \IEEEauthorblockA{Electrical and Computer Engineering Department\\
                    Purdue University\\
                    Email: elgamala@purdue.edu}
\and
  \IEEEauthorblockN{Petros Elia}
  \IEEEauthorblockA{EURECOM\\
                    Sophia Antipolis, France\\
                    Email: elia@eurecom.fr}
}
\maketitle

\begin{abstract}
\blfootnote{This work was supported by the ANR project ECOLOGICAL-BITS-AND-FLOPS.}
	In this work, we study a large linear interference network with an equal number of transmitters and receivers, where each transmitter is connected to two subsequent receivers. Each transmitter has individual access to a backhaul link (fetching the equivalent of $M_{T}$ files), while each receiver can cache a fraction $\gamma$ of the library. We explore the tradeoff between the communication rate, backhaul load, and caching storage by designing algorithms that can harness the benefits of cooperative transmission in partially connected networks, while exploiting the advantages of multicast transmissions attributed to user caching. We show that receiver caching and fetching content from the backhaul are two resources that can simultaneously increase the delivery performance in synergistic ways. 
	Specifically, an interesting outcome of this work is that user caching of a fraction $\gamma$ of the library can increase the per-user Degrees of Freedom (puDoF) by $\gamma$. Further, the results reveal significant savings in the backhaul load, even in the small cache size region. For example, the puDoF achieved using the pair $(M_{T}=8,
	\gamma=0)$ can also be achieved with the pairs $(M_{T}=4,\gamma=0.035)$ and $(M_{T}=2,\gamma=0.1)$, showing that even small caches can provide significant savings in the backhaul load.

	\end{abstract}

\section{Introduction}

The seminal work of \cite{maddah2014fundamental} showed that adding caches at the receivers can significantly reduce the delivery time of a communication network, by making it scalable to the number of users. Specifically, the work in \cite{maddah2014fundamental} studied the wired, single-stream, noiseless bottleneck channel where the transmitter has access to a library of $N$ files and serves the requests of $K$ receivers each equipped with a cache of size equal to $M$ files.
Using a novel pre-fetching and delivery method, the authors showed that the -- normalized -- delivery time of
\begin{align*}
    \mathcal{T}=\frac{K(1-\gamma)}{1+K\gamma}< \frac{1}{\gamma}
\end{align*}
can be achieved, which corresponds to each transmission serving a total of $D_{1}(\gamma)=K\gamma+1$ users ($\gamma\triangleq\frac{M}{N}\in[0,1]$).

The approach of \cite{maddah2014fundamental} was subsequently applied in other settings such as the wired, multi-server network\footnote{It can easily be shown that this network shares the same fundamental properties with the wireless $K$-user Multiple Input Single Output (MISO) Broadcast Channel (BC) with $L$ transmitting antennas.} with $L$ transmitting servers and $K$ receiving, cache-aided nodes \cite{7580630}, and the wireless, cache-aided interference channel  \cite{7857805} with $K_{T}$ transmitting nodes; each partially storing a fraction $\gamma_{T}$ of the library (where $K_{T}\gamma_{T}\triangleq L$). The surprising outcome of these works was that the number of users served per-transmission attributed to precoding, i.e., content being replicated at the transmitter side, and the corresponding number of users served due to \emph{coded transmissions}, i.e., due to content being replicated at the receiving nodes, appeared to be additive, achieving the delay of $\mathcal{T}=\frac{K(1-\gamma)}{L+K\gamma}$, which is translated to a sum Degrees of Freedom\footnote{The DoF are simply the delivery rate at high SNR (in units of \emph{file}, after normalization by $\log(\text{SNR})$). They reflect the total number of users served at a time and are calculated as the total number of information bits that need to be transmitted, divided by the delivery time i.e. $D=\frac{K(1-\gamma)}{T}$. In the same spirit, the per-user DoF are equal to the DoF divided by the number of users.} (DoF) performance equal to
\begin{align*}
    D_{L}(\gamma)=\frac{K(1-\gamma)}{\mathcal{T}}=L+K\gamma
\end{align*}
and thus, the per-user DoF become $d_{L}(\gamma)=\frac{1-\gamma}{\mathcal{T}}=\frac{L}{K}+\gamma.$


Contrarily, transmitter cooperation without receiver caching was shown to offer significant DoF gains in partially connected interference networks, through the use of Zero-Forcing (ZF) based schemes and a topology-aware choice of the downloaded messages at each transmitter~\cite{ElGamal-Veeravalli-Cambridge2018}. In particular, it was shown in~\cite{ElGamal-Veeravalli-ISIT14} and \cite{ElGamal-Annapureddy-Veeravalli-ICC12} that when each message can be downloaded from the backhaul (in each channel use) by an average of $M_T$ transmitters, and a maximum of $M_T$ transmitters, then the per-user DoF equals $\frac{4M_T-1}{4M_T}$ and $\frac{2M_T}{2M_T+1}$, respectively.

In this work, we focus on the setting with $K$ transmitters and $K$ receivers \cite{Wyner}, where each user $k\in\{0,1,...,K-1\}$ is receiving two interfering messages; one from transmitter $k-1$ and another from transmitter $k$.
Users are equipped with caches of normalized size $\gamma\in[0,1)$, and each will request one out of $N$ files. To serve these demands, transmitters are allowed to fetch from the backhaul, at most, the equivalent of $M_{T}$ files, while at the same time they are able to coordinate and perform cooperative transmissions.

\textit{Our interest lies in designing the caching, backhaul fetching and cooperative delivery algorithms that can jointly minimize the delivery time} of a single file to each user. The performance metric that we will use are the per-user DoF achieved upon complete delivery of all files.
 
As a basis for our scheme, we first consider the case with no caching at the receivers, and present a modification of the schemes in \cite{ElGamal-Veeravalli-ISIT14} and \cite{ElGamal-Annapureddy-Veeravalli-ICC12} that is tailored to our system model, i.e., we are allowed to divide a single file into subfiles, and consider the maximum per-transmitter backhaul load constraint required to deliver all files. Then, we add caches at the receivers and show how the insights from the cooperative transmission problem can be combined with the coded transmissions of the cache-aided literature (see \cite{maddah2014fundamental,7580630,7857805,shariatpanahi2017multi,lampiris2018resolving,8374958}) to further increase the DoF and provide savings in the backhaul.

\paragraph*{Related Work}

The work in \cite{7606851} considered a similar setting to ours and designed a caching and delivery policy that led to the characterization of the per-user DoF in large Wyner's networks. However, the authors in \cite{7606851} assumed that each transmitter can only download from the backhaul messages associated to the receivers connected to it. We relax this restriction here, by \textit{allowing transmitters to download any part of any file, as long as the backhaul constraint is respected}, and we show that this added flexibility will lead to superior performance.

Further, the work in \cite{8254141} considered a $K$-user partially connected interference network, where each receiver is connected to $L$ transmitters with succeeding indices, and caching is enabled at both transmitters and receivers. Contrary to the transmitter-side caching approach of \cite{8254141}, here the choice of downloaded content at each transmitter is based on receiver demands. Finally, the works in \cite{8007072,piovano2017generalized,7541657} consider cache-aided networks with imperfect or no Channel State Information at the Transmitters (CSIT). Here, we assume the availability of perfect CSIT, which enables the analysis of the limits of cooperative zero-forcing transmission strategies.

\section{System Model \& Notation}

We assume a set of $K$ transmitters, $\mathcal{K}_{T}\!\triangleq\!\{0, 1, ... K\!-\!1\}$, and a set of $K$ receivers, $\mathcal{K}_{R}\triangleq\{0, 1, ... K-1\}$, where transmitter $k\in\mathcal{K}_{T}$ is connected to receivers $k$ and $k+1$. The received signal at receiver $k+1$ is given by
\begin{equation*}
	y_{k+1}= x_{k+1}+h_{k,k+1}x_{k}+ w_{k+1},
\end{equation*}
where $x_{k}\in\mathbb{C}$ denotes the transmitted signal from transmitter $k$, that satisfies the average power constraint $\mathbb{E}\{\|x_{k}\|^{2}\}\le P$, $h_{k,k+1}\in\mathbb{C}$ denotes the channel realization between transmitter $k$ and receiver $k+1$, while $w_{k+1}$ corresponds to the channel noise, $w_{k+1}\thicksim \mathbb{C}\mathcal{N}(0,1)$.

We assume that the library $\mathcal{F}$ is comprized of $N$ files, each of size $f$ bits. Each receiver is equipped with a cache of $\gamma\cdot N\cdot f$ bits and will request one of the $N$ files, where $\gamma\in[0,1)$. We denote with $W^{r_{k}}\in\mathcal{F}$ the file requested by user $k$. Transmitters are connected by individual links to the backhaul and can each fetch $M_{T}\cdot f$ bits. Communication takes place in two phases. In the first phase, namely \emph{the placement phase}, the caches of the receivers are filled with content in a manner oblivious to future demands, but dependent on the network topology. Then, during the second phase, called the \emph{delivery} phase, each receiver requests one of the $N$ files and each transmitter downloads up to $M_T\cdot f$ bits through its backhaul link, where the backhaul downloads can be dependent on user demands.
The goal is to design the placement of content at the receivers, the fetching policy at the transmitters from the backhaul and the subsequent cooperative transmission in such a way so as to reduce the delivery time of a single file request to each user, for any pair\footnote{While the results are presented for specific pairs $(M_T, \gamma)$, they easily extend to \textit{any} pair using memory sharing, (see \cite{maddah2014fundamental} and App.~\ref{SecMemorySharing}).} $(M_T, \gamma)$. Upon complete delivery of all files, the considered performance metric is the asymptotic per-user Degrees of Freedom (puDoF) (i.e., the DoF normalized by the number of users in large networks), as defined in~\cite{ElGamal-Veeravalli-Cambridge2018}. We use $d(M_T,\gamma)$ to denote the puDoF achieved with a backhaul load $M_T$ and a fractional cache size $\gamma$.

\paragraph*{Notation}

The set of integers is denoted by $\mathbb{N}$. For $n,k\in\mathbb{N}$, we use $[n]_{k}\triangleq n\mod k$ to denote the modulo operation and $\binom{n}{k}$ for the $n$-choose-$k$ function, while for the product operation we follow the convention $\prod_{i=k}^{n}x_{i}=1$, if $k> n$. If ${A}$ is a set, we will denote its cardinality with $|{A}|$, while for any pair of sets $A,B$, we will use $A\setminus B$ to denote the difference set. Finally, we use the symbol $\oplus$ to denote the bit-wise XOR operation and $\mathcal{Z}_{k}$ to denote the cache content of receiver $k\in\mathcal{K}_{R}$. Using a slight abuse of notation, we will denote transmitted messages by the subfile(s) they contain.

\section{Main Results}

\begin{theorem}\label{thm:nocaching}
In the Wyner's network with per-transmitter maximum backhaul load $M_T\cdot f$ bits and no caches at the receivers, the per-user DoF for any $x\in\mathbb{N}$ satisfies the following:
\begin{align}\label{EqNonCache1}
	d\left(\frac{4x^2}{4x-1},\gamma=0\right)&=\frac{4x-1}{4x},\\
	d\left(\frac{x+1}{2},\gamma=0\right)&\geq\frac{2x}{2x+1}.
	\label{EqNonCache2}
\end{align}
\end{theorem}
\begin{proof}
     The proof of achievability is based on a modification of the schemes in \cite{ElGamal-Veeravalli-ISIT14} and \cite{ElGamal-Annapureddy-Veeravalli-ICC12}, and is provided in App.~\ref{sec:nocaching}.
	The converse of Eq.~\eqref{EqNonCache1} follows from \cite{ElGamal-Veeravalli-ISIT14}, where it was shown under an average backhaul load constraint. Since any scheme respecting a maximum load constraint is also respecting the average load constraint with the same value, it follows that the result is tight.
\end{proof}

\begin{theorem}\label{thm:caching}
In the Wyner's network with per-transmitter maximum backhaul load $M_{T}\cdot f$ bits and a normalized cache size at each receiver of a fraction $\gamma$ of the library, the per-user DoF of $d(M_{T},\gamma)=1$ can be achieved with the following pairs for any $x\in\mathbb{N}$:
\begin{align}\label{EqCaching1}
	d&\left(\frac{1-\gamma^2}{4\gamma},\frac{1}{2x+1}\right)=1,\\
	d&\left(\frac{1}{4\gamma},\frac{1}{2x}\right)=1.\label{EqCaching2}
\end{align}
\end{theorem}
\begin{proof}
	The proof is constructive and presented in Sec. \ref{sec:caching}.
\end{proof}
\begin{corollary}\label{propAdditiveGamma}
    Caching a fraction $\gamma\!=\!\frac{1}{4x},x\in\mathbb{N}$ of the library at the receivers can increase the puDoF by an additive factor $\gamma$, while simultaneously decreasing the backhaul load by a multiplicative factor of $1-\gamma$.
\end{corollary}
\begin{proof}
    The proof makes use of the results from Eq.~\eqref{EqNonCache1} and Eq.~\eqref{EqCaching2}. 
    Starting from a backhaul load of $M_{T}=\frac{4x^2}{4x-1}$, and adding a fractional cache size of $\gamma=\frac{1}{4x}$ at each receiver, the new backhaul load becomes $M'_{T}=\frac{4x-1}{4x}\frac{4x^{2}}{4x-1}=x$. We conclude the proof by observing through Eq.~\eqref{EqCaching2} that the pair $(M'_{T},\gamma)=(x,\frac{1}{4x})$ leads to achieving the full puDoF.
\end{proof}

\begin{remark}
Observing the result in \cite[Theorem $1$]{7606851}, we can see that in order to achieve complete interference mitigation, it is required to have a backhaul load of $M_{T}=2$ and a cache size of $\gamma=\frac{1}{6}$. On the contrary, here we can achieve the maximal puDoF with the backhaul - caching pairs $(M_{T}=2, \gamma=\frac{1}{8})$ and $(M_{T}=\frac{3}{2}, \gamma=\frac{1}{6})$.

\end{remark}
The key factor enabling our result is that we allow for a more flexible backhaul load, instead of restricting each transmitter to download a specific set of messages which, in turn, allows to utilize transmitter cooperation more efficiently.

\section{Placement and Delivery of Files with Caching at the Receivers}\label{sec:caching}

In this section, we describe the scheme leading to the result of Theorem \ref{thm:caching}. We provide the proof of Eq.~\eqref{EqCaching1}, i.e., when the cache size takes values $\gamma=\frac{1}{2x+1},~x\in\mathbb{N}$, while noting that Eq.~\eqref{EqCaching2} would follow by using memory sharing (cf. App.~\ref{SecMemorySharing}).

\subsubsection{Placement Phase}

In the placement phase, each file is subpacketized into $S={1}/{\gamma}$ subfiles, i.e., for every file $W^{n}\in\mathcal{F}$ we have $W^{n}\!\to\! \{W_{0}^{n},  ..., W_{S\!-\!1}^{n}\}$. Users cache according to
\begin{align}\label{EqCacheContent}
	\mathcal{Z}_{k}=\left\{ W^{n}_{[k]_{S}}, ~\forall n\in \{1,2,...,N\}\right\}.
\end{align}

\subsubsection{Delivery Phase}

\begin{algorithm}[t!]\caption{Delivery Phase of the Cache-aided Scheme}\label{alg:DeliveryCaching}
	\For{	$m\in\{1, ..., \frac{1-\gamma}{2\gamma}\}$ (Choose a Delivery Network)}
	{
	\For{$p\in\{0,m\} $ (Choose Slot of Delivery Net)}{
		Transmitter $k:~0\le[k+p]_{2m}< m$ sends:
		\begin{align*}
			x_{k}\!=\!\sum_{i=0}^{[k]_{2m}}\!(-1)^{i} \!\prod_{j=k\!-\!i}^{k-1}h_{j,j+1} W_{[k+m\!-\!i]_{S}}^{r_{k-i}}\oplus W_{[k-i]_{S}}^{r_{k\!+\!m\!-\!i}}
		\end{align*}\\
		Transmitter $k:~m\le[k+p]_{2m}<2m-1$ sends:
		\begin{align*}
		x_{k}\!=\!\sum_{i=[k\!+\!1]_{m}}^{m-1}\!(\!-\!1)^{i}\! \prod_{j\!=\!k-i}^{k\!-\!1}h_{j,j\!+\!1} W_{[k\!+\!m\!-\!i]_{S}}^{r_{k-i}}\!\oplus\! W_{[k-i]_{S}}^{r_{k\!+\!m\!-\!i}}	
		\end{align*}\\
		Transmitter $k:~[k+p]_{2m}=2m-1$ sends:
		\begin{align*}
			x_{k}=\emptyset.
		\end{align*}
		}
	}
\end{algorithm}

As discussed above, the delivery phase starts with the request from each user of \emph{any}\footnote{We will assume that each user requests a different file, which corresponds to the worst case user demand.} file from the library $\mathcal{F}$. For $\gamma=\frac{1}{2x+1},~x\in\mathbb{N}$, the goal is to rely on the smallest possible backhaul load that can allow for interference-free reception i.e., $d(M_{T},\gamma)=1$.

The delivery phase consists of $2x$ transmission slots, where in each slot, we deliver a fraction $\gamma=\frac{1}{2x+1}$ of the requested file to every receiver which, along with the cached fraction will amount to the whole file. We will call each successive pair of delivery slots a \textit{Delivery Network} ($DN_{m},~m\in\{1,2,...,x\}$). Thus, there will be a total of $x$ delivery networks. The role of $DN_{m}$ is to deliver to user $k\in\mathcal{K}_{R}$ subfiles indexed by $[k\pm m]_{S}$.
To this end, during $DN_{m}$ the transmitted messages contain XORs (or linear combinations of XORs), where each XOR is formed using two subfiles with difference of indices equal to $[m]_{S}$. For example, in $DN_{m},~m\in\{1,...,\frac{1-\gamma}{2\gamma}\}$, the two transmitted XORs, intended for user $k\in\mathcal{K}_{R}$, will be
\begin{align*}
	W^{r_{k}}_{[k+m]_{S}}\oplus W^{r_{k+m}}_{[k]_{S}},\ \ \ \ \ \ \ 
	W^{r_{k}}_{[k-m]_{S}}\oplus W^{r_{k-m}}_{[k]_{S}}.
\end{align*}
Transmission takes place according to Alg.~\ref{alg:DeliveryCaching}.
First, we demonstrate how the algorithm succeeds in achieving full puDoF through the following example and subsequently we discuss the mechanics of Alg.~\ref{alg:DeliveryCaching}.

\begin{example}\label{exampleScheme}
	Let us assume that each user can store a fraction $\gamma\!=\!\frac{1}{5}$ of the library, which corresponds to $4$ transmission slots and thus $2$ Delivery Networks, namely $DN_{1}$ and $DN_{2}$. We begin by subpacketizing each file into $5$ subfiles and caching at each user according to Eq.~\eqref{EqCacheContent}, i.e.,
	\begin{align*}
		\mathcal{Z}_{0}=&\left\{W^{n}_{0},~\forall n\in \{1,2,...,N\}\right\},\\
		\mathcal{Z}_{1}=&\left\{W^{n}_{1},~\forall n\in \{1,2,...,N\}\right\},\\
		\vdots\\
		\mathcal{Z}_{5}=&\left\{W^{n}_{0},~\forall n\in \{1,2,...,N\}\right\}.
	\end{align*}
	
		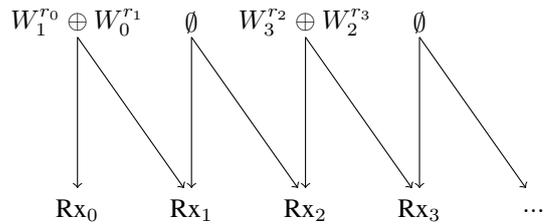
\begin{figure}[thb]
\centerline{
    \begin{tikzpicture}
    \foreach \i in {0,1.5,3,4.5}
    {
    	\draw [->] (10+\i,10) -- (10+\i,8);
    	\draw [->] (10+\i,10) -- (11.4+\i,8);	
	}
	\draw node at (10,10.2) { $W^{r_{0}}_{1}\oplus W^{r_{1}}_{0}$};
	\draw node at (13,10.2) { $W^{r_{2}}_{3}\oplus W^{r_{3}}_{2}$};
	\draw node at (11.5,10.2) { $\emptyset$};
	\draw node at (14.5,10.2) { $\emptyset$};
	\draw node at (10,7.7) { Rx$_{0}$};
	\draw node at (11.5,7.7) { Rx$_{1}$};
	\draw node at (13,7.7) { Rx$_{2}$};
	\draw node at (14.5,7.7) { Rx$_{3}$};
	\draw node at (16,7.7) { $...$};
    \end{tikzpicture}
}
\caption{Slot 1 of Delivery Network $DN_{1}$.}
\label{fig:DeliveryNet1Slot1}
\end{figure}
\begin{figure}[h]
\centerline{
    \begin{tikzpicture}
    \foreach \i in {0,1.5,3,4.5}
    {
    	\draw [->] (10+\i,10) -- (10+\i,8);
    	\draw [->] (10+\i,10) -- (11.4+\i,8);	
	}
	\draw node at (11.5,10.2) { $W^{r_{2}}_{1}\oplus W^{r_{1}}_{2}$};
	\draw node at (14.5,10.2) { $W^{r_{4}}_{3}\oplus W^{r_{3}}_{4}$};
	\draw node at (10,10.2) { $\emptyset$};
	\draw node at (13,10.2) { $\emptyset$};
		\draw node at (10,7.7) { Rx$_{0}$};
	\draw node at (11.5,7.7) { Rx$_{1}$};
	\draw node at (13,7.7) { Rx$_{2}$};
	\draw node at (14.5,7.7) { Rx$_{3}$};
	\draw node at (16,7.7) { $...$};
    \end{tikzpicture}
}
\caption{Slot 2 of Delivery Network $DN_{1}$.}
\label{fig:DeliveryNet1Slot2}
\end{figure}
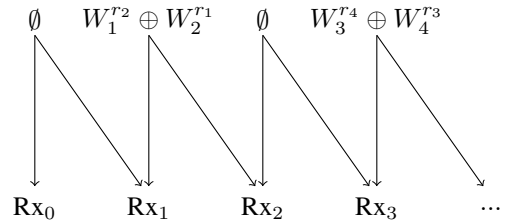
	

	After the request of a single file from each user, the transmission begins with $DN_{1}$ and then with $DN_{2}$. The first pair of transmission slots are responsible for delivering XORs comprized of subfiles with subsequent indices i.e., $W^{r_{0}}_{1}\oplus W^{r_{1}}_{0}$, $W^{r_{1}}_{2}\oplus W^{r_{2}}_{1}$, $W^{r_{2}}_{3}\oplus W^{r_{3}}_{2}$, $W^{r_{3}}_{4}\oplus W^{r_{4}}_{3}$ and so on. The transmitted messages at the first $4$ transmitters during $DN_{1}$ are illustrated in Fig. \ref{fig:DeliveryNet1Slot1}-\ref{fig:DeliveryNet1Slot2}.

The two slots of $DN_{2}$ follow after the completion of the two slots of Delivery Network $DN_{1}$. Here, the transmitters will communicate subfiles to user $k$ with indices $[k\pm2]_{S}$ i.e., $W^{r_{0}}_{2}\oplus W^{r_{2}}_{0}$, $W^{r_{1}}_{3}\oplus W^{r_{3}}_{1}$, $W^{r_{2}}_{4}\oplus W^{r_{4}}_{2}$, $W^{r_{3}}_{0}\oplus W^{r_{5}}_{3}$, and so on. The transmitted messages for each of the two slots are illustrated in Fig.~\ref{fig:DeliveryNet2Slot1}-\ref{fig:DeliveryNet2Slot2}.

In each of the $4$ slots from Delivery Networks $DN_{1}$ and $DN_{2}$, a different subfile is delivered to each receiver, thus completing the delivery of all files\footnote{We note that while $W^{r_0}$ is not completely delivered, asymptotically that does not affect the per-user DoF of a large network.}, while downloading exactly $6$ subfiles at each transmitter.
\end{example}

\paragraph*{Details of the Delivery Algorithm}
First, a delivery network is chosen (Step 1), and then one of the two slots of the delivery network is chosen (Step 2).
As discussed above, the purpose of delivery network $DN_{m}$ is to deliver to each receiver $k\in\mathcal{K}_{R}$ subfiles indexed as $[k \pm m]_{S}$. During each transmission slot, the transmitters are divided into three non-overlapping sets. The first set (Line $3$) is tasked with transmitting new messages and nulling the interference created by the messages of previous transmitters. The second set (Line $4$) is tasked with transmitting messages that nullify the interference at their respective receiver, which interfering messages have been generated by transmitters of the first set. Finally, the third set (Line $5$) of transmitters remains silent.

\subsection*{Characterizing the Required Backhaul Load}

In this section, we will characterize the backhaul load that our algorithm requires in order to achieve interference-free transmission for a given fractional cache $\gamma=\frac{1}{2x+1},~x\in\mathbb{N}$.

We begin by observing (cf. Alg. \ref{alg:DeliveryCaching} and Ex.~\ref{exampleScheme}) that the backhaul load at each transmitter during a specific Delivery Network is -- potentially -- different, and the two slots of a delivery network are designed to balance the per-transmitter backhaul load. As an example, in Fig.~\ref{fig:DeliveryNet2Slot1}-\ref{fig:DeliveryNet2Slot2} we can see that if a transmitter is silent during one slot of $DN_{2}$, then during the other slot it will transmit the linear combination of two XORs, thus will need to fetch from the backhaul $4$ subfiles.

Consider a transmitter that, during Slot $2$ of $DN_{m}$, is silent. This transmitter's index, $k$, (Line 5 of Alg.~\ref{alg:DeliveryCaching}) must satisfy $[k+m]_{2m}=2m-1$, which gives $k=(2b-1)m-1,~b\in\mathbb{N}$. This further means that during Slot $1$ of $DN_{m}$, the transmitter's load will be characterized by Line $3$ of Alg.~\ref{alg:DeliveryCaching}, since $[k]_{2m}=[2bm- m-1]_{2m}=m-1$. Thus, this transmitter will need to fetch the contents of $m$ XORs, making the total, per-delivery-network, backhaul load equal to $2m$ subfiles.
Using this observation, we can calculate the overall required per-transmitter backhaul load, which is (note:  $x=\frac{1-\gamma}{2\gamma}$)
\begin{align*}
	M_{T}=\frac{1}{S}\cdot\sum_{m=1}^{x}2m=\gamma\cdot 2\cdot \frac{x(x+1)}{2}=\frac{1-\gamma^{2}}{4\gamma}.
\end{align*}

\begin{figure}[ht!]
\centering
\includegraphics[width=0.85\columnwidth]{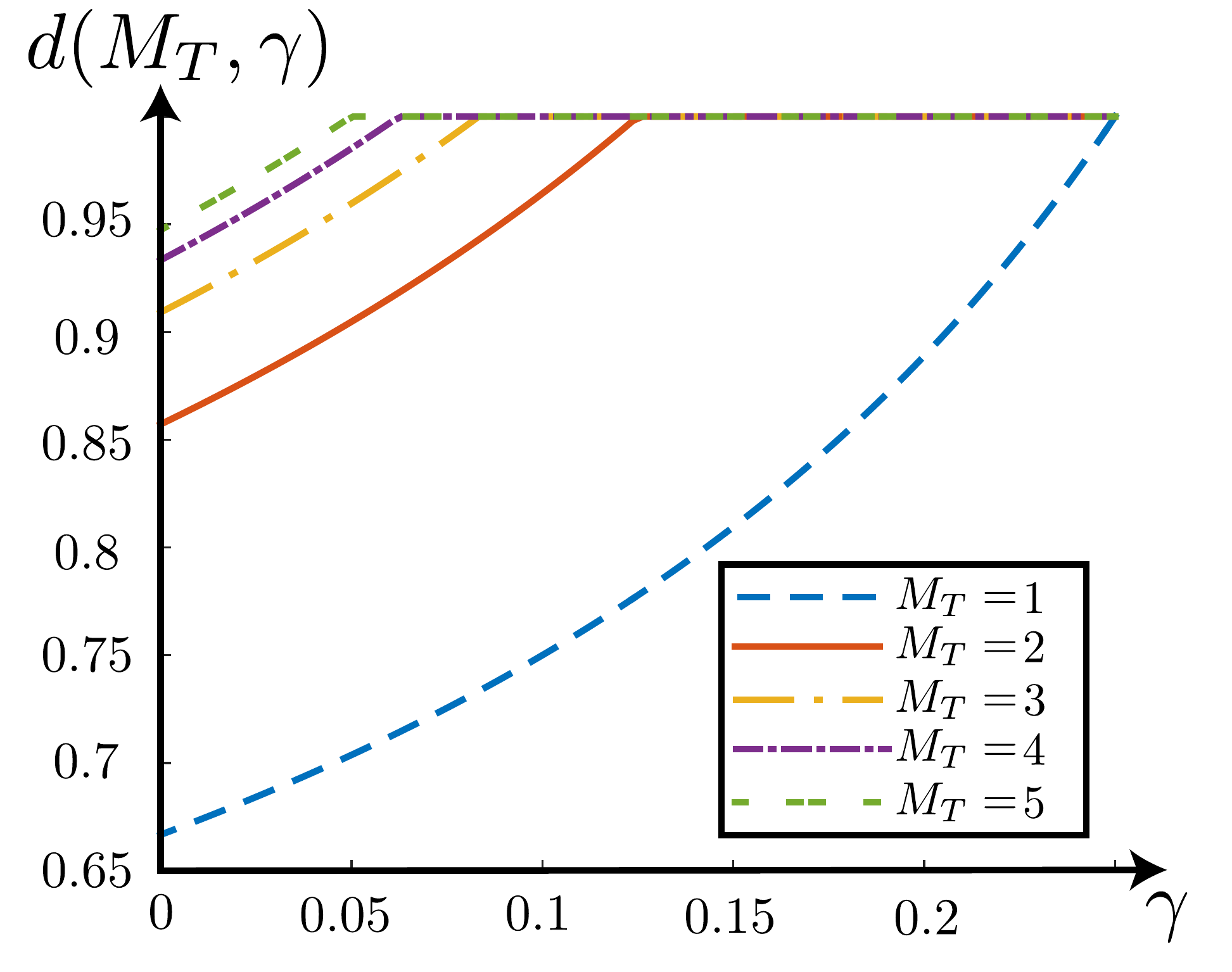}
\caption{Per-user DoF as a function of cache size for different values of backhaul load.}
\label{fig:pudofWithCaching}
\end{figure}
\begin{figure}[ht!]
\centering
\includegraphics[width=0.87\columnwidth]{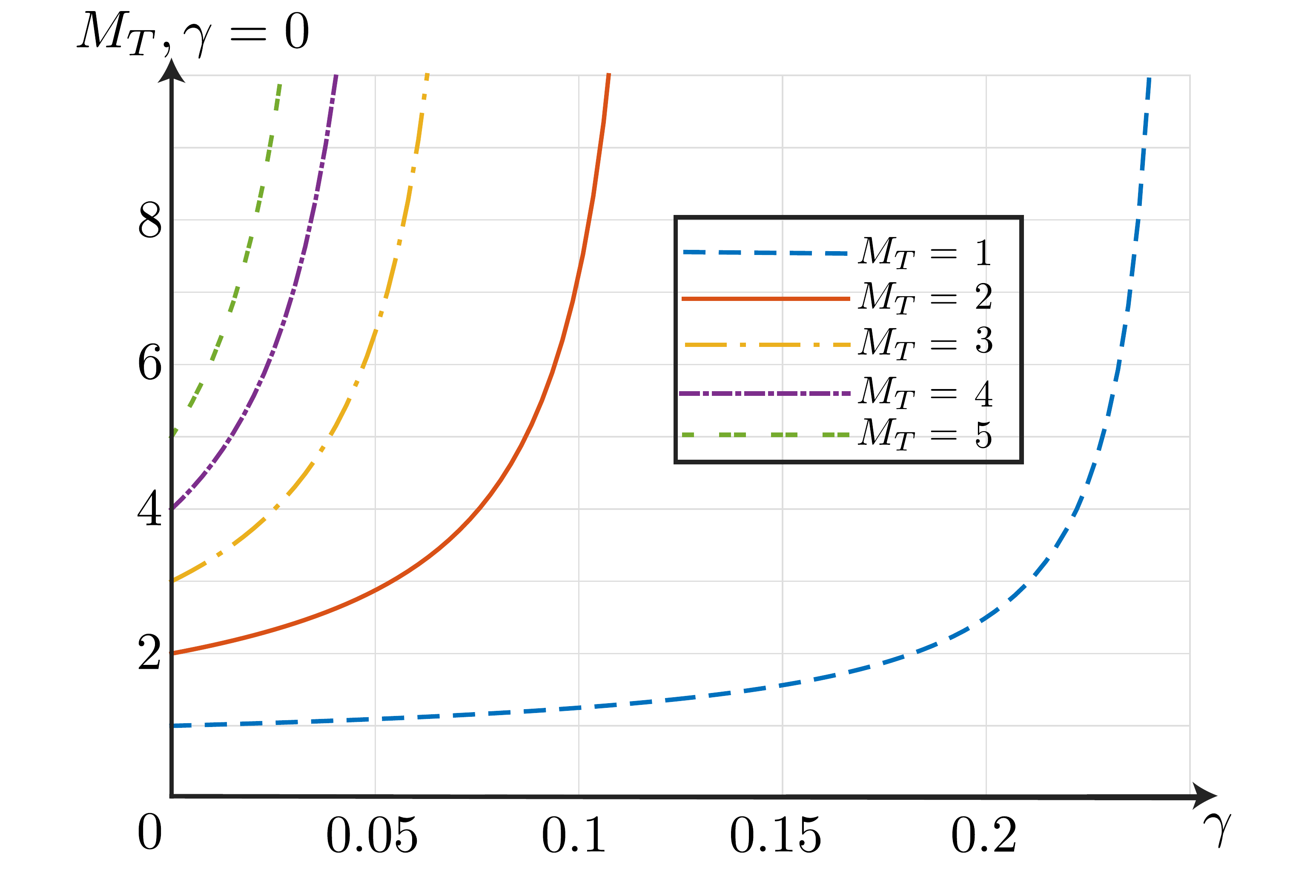}
\caption{Required backhaul load without user caching that achieves the same per-user DoF as the cache-aided scheme with the pair $(M_{T},\gamma)$.}
\label{fig:requiredBackhaulWithoutCaching}
\end{figure}

\section{Discussion and Concluding Remarks}

From Corollary~\ref{propAdditiveGamma}, we can deduce that receiver-side caching impacts the delivery time in three different ways.
\begin{figure*}[t!]
\centerline{
    \begin{tikzpicture}
    \foreach \i in {0,2,4,6, 8, 10, 12}
    {
    	\draw [->] (10+\i,10) -- (10+\i,8);
    	\draw [->] (10+\i,10) -- (11.5+\i,8);	
	}
	\draw node at (9.8,10.6) { $W^{r_{0}}_{2}\!\oplus\! W^{r_{2}}_{0}$};
	\draw node at (14.4,10.6) { $-h_{1,2}(W^{r_{1}}_{3}\!\oplus\! W^{r_{3}}_{1})$};
	\draw node at (12,11) { $-h_{0,1}(W^{r_{2}}_{0}\!\oplus\! W^{r_{0}}_{2})$};
	\draw node at (12,10.2) {$W^{r_{1}}_{3}\!\oplus\! W^{r_{3}}_{1}$};
	\draw node at (16,10.5) { $\emptyset$};
	\draw node at (11.1,9.2) {$h_{0,1}$};
	\draw node at (13.1,9.2) {$h_{1,2}$};
	\draw node at (15.1,9.2) {$h_{2,3}$};
	\draw node at (17.1,9.2) {$h_{3,4}$};
	\draw node at (19.1,9.2) {$h_{4,5}$};
	\draw node at (21.1,9.2) {$h_{5,6}$};
	\draw node at (23.1,9.2) {$...$};
	\draw node at (10,7.7) { Rx$_{0}$};
	\draw node at (12,7.7) { Rx$_{1}$};
	\draw node at (14,7.7) { Rx$_{2}$};
	\draw node at (16,7.7) { Rx$_{3}$};
	\draw node at (18,7.7) { Rx$_{4}$};
	\draw node at (20,7.7) { Rx$_{5}$};
	\draw node at (22,7.7) { Rx$_{6}$};
	\draw node at (23.7,7.7) {$...$};
	\draw node at (17.9,10.6) { $W^{r_{4}}_{1}\!\oplus\! W^{r_{6}}_{4}$};
	\draw node at (22.4,10.6) { $-h_{5,6}(W^{r_{5}}_{2}\!\oplus\! W^{r_{7}}_{0})$};
	\draw node at (20,11) { $-h_{4,5}(W^{r_{4}}_{1}\!\oplus\! W^{r_{6}}_{4})$};
	\draw node at (20,10.2) {$W^{r_{5}}_{2}\!\oplus\! W^{r_{7}}_{0}$};
    \end{tikzpicture}
}
\caption{Slot 1 of Delivery Network $DN_{2}$.}\label{fig:DeliveryNet2Slot1}
\end{figure*}
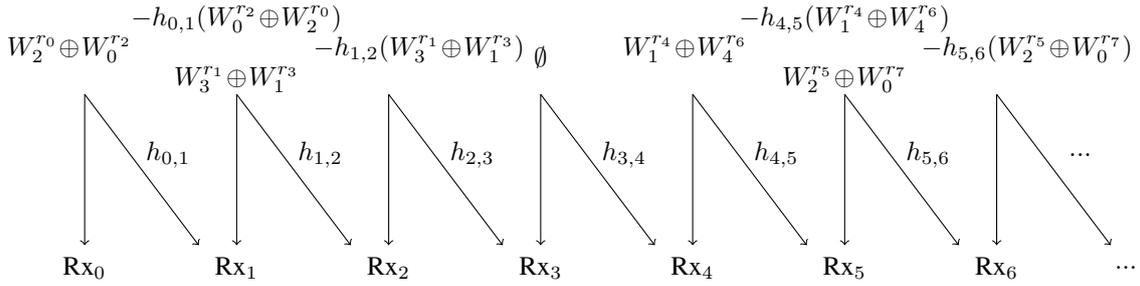
\begin{figure*}[t!]
\centerline{
    \begin{tikzpicture}
    \foreach \i in {0,2,4,6, 8, 10, 12}
    {
    	\draw [->] (10+\i,10) -- (10+\i,8);
    	\draw [->] (10+\i,10) -- (11.5+\i,8);	
	}
	\draw node at (9.8,10.6) { $W^{r_{1}}_{4}$};
	\draw node at (13.7,10.6) { $W^{r_{2}}_{4}\!\oplus\! W^{r_{4}}_{2}$};
	\draw node at (12,10.6) { $\emptyset$};
	\draw node at (15.7,11) { $-h_{2,3}(W^{r_{2}}_{4}\!\oplus\! W^{r_{4}}_{2})$};
	\draw node at (15.7,10.2) { $W^{r_{3}}_{0}\!\oplus\! W^{r_{5}}_{3}$};
	\draw node at (11.1,9.2) {$h_{0,1}$};
	\draw node at (13.1,9.2) {$h_{1,2}$};
	\draw node at (15.1,9.2) {$h_{2,3}$};
	\draw node at (17.1,9.2) {$h_{3,4}$};
	\draw node at (10,7.7) { Rx$_{0}$};
	\draw node at (12,7.7) { Rx$_{1}$};
	\draw node at (14,7.7) { Rx$_{2}$};
	\draw node at (16,7.7) { Rx$_{3}$};
	\draw node at (18,7.7) { Rx$_{4}$};
	\draw node at (20,7.7) { Rx$_{5}$};
	\draw node at (22,7.7) { Rx$_{6}$};
	\draw node at (23.7,7.7) {$...$};
	\draw node at (18.1,10.6) { $-h_{3,4}(W^{r_{3}}_{0}\!\oplus\! W^{r_{5}}_{3})$};
	\draw node at (22,10.6) { $W^{r_{6}}_{3}\!\oplus\! W^{r_{8}}_{1}$};
	\draw node at (20,10.6) { $\emptyset$};
	\draw node at (19.1,9.2) {$h_{4,5}$};
	\draw node at (21.1,9.2) {$h_{5,6}$};
	\draw node at (23.1,9.2) {$...$};
    \end{tikzpicture}
}
\caption{Slot 2 of Delivery Network $DN_{2}$.}\label{fig:DeliveryNet2Slot2}
\end{figure*}
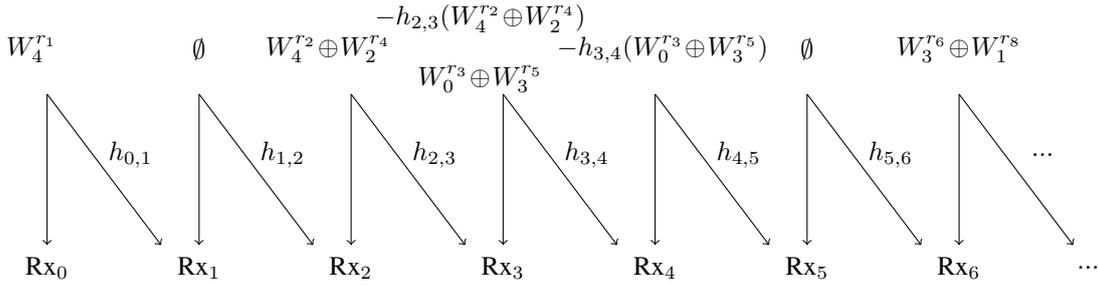

	\paragraph{Local Caching Gain} Having stored a fraction $\gamma$ from each of the files, the system can have reductions in the delivery time since part of the desired content is already stored at the receivers and hence it is not required to be communicated.
	\paragraph{Multicasting Gain} Since messages contain XORed subfiles, in order to decode its desired subfile each receiver needs to make use of its cached but unwanted content. Thus, unwanted, cached content allows the transmission of more than one message simultaneously, which saves transmission slots.
	\paragraph{Cooperative Transmission Gain} As a fraction $\gamma$ of each file is cached at each receiver, the user will require only the smaller fraction $(1-\gamma)$ of the file. Now, for the same backhaul load as the no-caching case, this smaller request (from $1$ to $1-\gamma$) permits the transmitters to fetch more content, which can further boost the cooperation gains.

In Fig. \ref{fig:pudofWithCaching}, we illustrate the above points b) and c) by plotting the puDoF that is achieved using different pairs $(M_T, \gamma)$. It is interesting to note that a high backhaul load paired with a small cache can provide an intereference-free reception at every node.

Further, in Fig. \ref{fig:requiredBackhaulWithoutCaching}, we plot\footnote{The $M_{T}$ values of the $y$-axis are calculated according to the results of Theorem~\ref{thm:nocaching}. While not all of the points presented may be achievable, nevertheless their convex envelope is, and as a result present an even more optimistic case in favor of the no-caching schemes.} the backhaul load that would have been needed to achieve the same per-user DoF as does the pair $(M_{T}, \gamma)$. We can note here that caching even a fraction $\gamma=\frac{1}{20}$ with a backhaul load of $M_{T}=3$ would have otherwise required a no caching backhaul load of $M_{T}=6$. Moreover, caching a fraction $\gamma=0.1$ can reduce the load from $M_{T}\approx 7$ to $M_T=2$. This further accentuates the role of coded transmissions and multicasting as relevant and impactful techniques that allow for fast delivery of content.

On the other hand, in the absence of caching, the cost of increasing the DoF even by a small fraction would have been extremely high. For example, if $M_{T}=\frac{16}{7}$ we know that we can achieve $d=\frac{7}{8}$, but in order to achieve $d'=\frac{15}{16}$, we would have to more than double the backhaul cost (see Eq.~\eqref{EqNonCache1}). Contrarily, the same increase can be achieved by caching at each user an (approximate) fraction $\gamma=\frac{1}{16}$ of the library, and requiring a backhaul load of only $M_T=2$.

To summarize, in this work we have characterized the per-user DoF with no caching for all interger values of the backhaul load, as well as other rational values. We also characterized the backhaul load required for complete interference mitigation with fractional receiver cache sizes that are equal to $\frac{1}{x}$ for every integer value $x>1$. We have demonstrated through the obtained results the effectiveness of receiver caching in the studied setting and how it leads to significant savings in both the delivery time and required backhaul load. Further, we have demonstrated how a flexible allocation of messages over the backhaul leads to significant reductions in both the backhaul load and cache size needed to completely mitigate interference from a DoF perspective.


%
 \appendices

\section{No-Caching Schemes}\label{sec:nocaching}

In this section, we provide the achievable schemes, in the absence of caching, that prove Theorem~\ref{thm:nocaching}. The presented schemes rely on applications of the schemes in~\cite{ElGamal-Veeravalli-ISIT14} and \cite{ElGamal-Annapureddy-Veeravalli-ICC12}\footnote{See also \cite{ElGamal-Veeravalli-Cambridge2018} for a summary and high-level illustration.} with time sharing, in order to meet the considered backhaul constraint.

\begin{algorithm}[tbh]\caption{Delivery Phase Under no Caching corresponding to the result of Eq.~\eqref{EqNonCache1}}\label{alg:DeliveryNoCache1}
Assume $M_{T}=\frac{4x^2}{4x-1},~x\in\mathbb{N}$.\\
Next($i$) returns the smallest index of a subfile of $W^{r_i}$ that has not been transmitted in a previous time slot.\\
Subpacketize each file into $S=4x-1$ subfiles.\\
	\For{	$t\in\{0,1, ..., 4x-1\}$ (Time Slots)}
	{
		Transmitter $k:~0\le[k-t]_{4x}< 2x$ sends:
		\begin{align*}
				x_{k}\!=\!\sum_{i=0}^{[k-t]_{4x}}(-1)^{i} \left(\prod_{j=k-i}^{k-1}h_{j,j+1}\right) W_{\text{Next}(k-i)}^{r_{k-i}}.
		\end{align*}\\
		Transmitter $k:~2x\le[k-t]_{4x}<4x-1$ sends:
		\begin{align*}
				x_{k}=\sum_{i=1}^{2x-L} (-1)^{i-1} \left(\prod_{j=k+1}^{k+i-1}\frac{1}{h_{j-1,j}}\right) W_{\text{Next}(k+i)}^{r_{k+i}},
		\end{align*}
		where $L=[k-t]_{4x}-2x-1$.\\
		Transmitter $k:~[k-t]_{4x}= 4x-1$ does not transmit.
  }
\end{algorithm}
\subsection{Proof of Theorem~\ref{thm:nocaching}, Eq. \eqref{EqNonCache1}}
The proposed scheme is completed in $4x$ blocks of communication. Each receiver gets $4x-1$ packets, and each packet is delivered through a one degree of freedom link. Each file is subpacketized into $4x-1$ subfiles, i.e. file $W^{n}\to \{W_{0}^{n}, W_{1}^{n}, ..., W_{4x-2}^{n}\}$. Each subfile is carried over one packet.
In each block of communication, we divide the network into subnetworks, each consisting of $4x$ consecutive transmitter-receiver pairs, and use the scheme in~\cite{ElGamal-Veeravalli-ISIT14} to deliver $4x-1$ packets in each subnetwork.

In what follows, we explain the scheme for the case when $x=1$ for simplicity, and demonstrate how it generalizes to larger values of $x$.
For the first block of communication when $x=1$, the first transmitter downloads $W_{0}^{r_{0}}$, the second transmitter downloads $W_{0}^{r_{0}}$ and $W_{0}^{r_{1}}$, and the third downloads $W_{0}^{r_{3}}$. All three subfiles $W_{0}^{r_{0}}, W_{0}^{r_{1}}$ and $W_{0}^{r_{3}}$ can then be delivered through one DoF links using cooperative transmission, as illustrated in~\cite{ElGamal-Veeravalli-ISIT14}. The fourth transmitter is inactive, thereby eliminating inter-subnetwork interference, and thus allowing the same scheme to be applied to each remaining subnetwork.

In the second block of communication, the first transmitter is inactive, while the same scheme is applied while we allocate to the second transmitter the role that the first transmitter had in the network, to the third transmitter the role that the second transmitter had and so on. More precisely, the first subnetwork would now consist of users with indices $\{1,2,3,4\}$. Subfile $W_{1}^{r_{1}}$ would then be downloaded by transmitters $1$ and $2$, $W_{0}^{r_{2}}$ would be downloaded by transmitter $2$, and $W_{1}^{r_{4}}$ would be downloaded by transmitter $3$, while the fifth transmitter is deactivated. Since transmitter $5$ is inactive, then there is no inter-subnetwork interference, hence the same scheme can be applied to the subnetwork containing users $\{5,6,7,8\}$ to deliver subfiles $\{W_{1}^{r_{5}}, W_{0}^{r_{6}}, W_{1}^{r_{8}} \}$ without causing interference to the following subnetwork, and similarly, three subfiles can be delivered over one DoF links for every subsequent subnetwork.   
Proceeding in a similar fashion as above for the third block of communication for the case when $x=1$, we are able to deliver all $4x-1=3$ subfiles of each file $W^i$ for almost all files\footnote{In fact, we deliver all files $W_i^{r_{k}}$, whose index $i \geq 4x-1$, but since the focus is on the asymptotic puDoF, then ignoring a small set of users would not affect the result.} in the $4x=4$ communication blocks. The achieved puDoF would then be given by $\frac{4x-1}{4x}=\frac{3}{4}$. For the backhaul load, in each block of communication, each transmitter downloads an average of $x$ subfiles. Over the $4x$ communication blocks, each transmitter downloads $4x^2$ subfiles, resulting in $M_T=\frac{4x^2}{4x-1}$ files. 
%

\begin{algorithm}[tbh]\caption{Delivery Phase Under no Caching corresponding to the result of Eq.~\eqref{EqNonCache2}}\label{alg:DeliveryNoCache}
Assume $M_{T}=\frac{x+1}{2},~x\in\mathbb{N}$.\\
Next($i$) returns the smallest index of a subfile of $W^{r_i}$ that has not been transmitted in a previous time slot.\\
Subpacketize each file into $S=2x$ subfiles.\\
	\For{	$t\in\{0,1, ..., 2x\}$ (Time Slots)}
	{
		Transmitter $k:~0\le[k-t]_{2x+1}< x$ sends:
		\begin{align*}
				x_{k}\!=\!\sum_{i=0}^{[k-t]_{2x+1}}(-1)^{i} \left(\prod_{j=k-i}^{k-1}h_{j,j+1}\right) W_{\text{Next}(k-i)}^{r_{k-i}}.
		\end{align*}\\
		Transmitter $k:~x\le[k-t]_{2x+1}<2x$ sends:
		\begin{align*}
				x_{k}=\sum_{i=1}^{x-L} (-1)^{i-1} \left(\prod_{j=k+1}^{k+i-1}\frac{1}{h_{j-1,j}}\right) W_{\text{Next}(k+i)}^{r_{k+i}},
		\end{align*}
		where $L=[k-t]_{2x+1}-x$.\\
		Transmitter $k:~[k-t]_{2x+1}= 2x$ does not transmit.
  }
\end{algorithm}
\begin{figure*}[t!]
\centerline{
    \begin{tikzpicture}
    \foreach \i in {0,2,4,6, 8, 10, 12, 14,15.7}
    {
    	\draw [->] (10+\i,10) -- (10+\i,8);
    	\draw [->] (10+\i,10) -- (11.5+\i,8);	
	}
	\draw node at (9.8,10.6) { $W^{r_{0}}_0$};
    \draw node at (11.8,10.6) { $W^{r_{0}}_0, W^{r_{1}}_0$};
	\draw node at (13.7,11) { $W^{r_{0}}_0,W^{r_{1}}_0$};
    \draw node at (13.9,10.2) { $W^{r_{2}}_0$};
	\draw node at (15.7,11) { $W^{r_{0}}_0, W^{r_{1}}_0$};
	\draw node at (15.7,10.2) { $W^{r_{2}}_0,W^{r_{3}}_0$};
	\draw node at (11.1,9.2) {$h_{0,1}$};
	\draw node at (13.1,9.2) {$h_{1,2}$};
	\draw node at (15.1,9.2) {$h_{2,3}$};
	\draw node at (17.1,9.2) {$h_{3,4}$};
	\draw node at (10,7.7) { Rx$_{0}$};
	\draw node at (12,7.7) { Rx$_{1}$};
	\draw node at (14,7.7) { Rx$_{2}$};
	\draw node at (16,7.7) { Rx$_{3}$};
	\draw node at (18,7.7) { Rx$_{4}$};
	\draw node at (20,7.7) { Rx$_{5}$};
	\draw node at (22,7.7) { Rx$_{6}$};
	\draw node at (24,7.7) { Rx$_{7}$};
    \draw node at (25.7,7.7) {Rx$_{8}$};
	\draw node at (27.4,7.7) {$...$};
    \draw node at (18.1,11) { $W^{r_{5}}_0, W^{r_{6}}_0$ };
 	\draw node at (18.1,10.2) {$W^{r_{7}}_0, W^{r_{8}}_0$};
 	\draw node at (20.4,10.7) {$W^{r_{6}}_0, W^{r_{7}}_0,W^{r_{8}}_0$};
 	\draw node at (22,10.2) { $W^{r_{7}}_0,W^{r_{8}}_0$};
 	\draw node at (24,10.2) { $W^{r_{8}}_0$};
 	\draw node at (25.7,10.2) {$\emptyset$};
	\draw node at (19.1,9.2) {$h_{4,5}$};
	\draw node at (21.1,9.2) {$h_{5,6}$};
	\draw node at (23.1,9.2) {$h_{6,7}$};
   	\draw node at (25.1,9.2) {$h_{7,8}$};
    \draw node at (27.1,9.2) {$...$};
    \end{tikzpicture}
}
\caption{One transmission slot in the non-cache-aided network with backhaul constraint of $M_{T}=\frac{5}{2}$ corresponding to Eq. \eqref{EqNonCache2}. This slot involves a subnetwork of $9$ users and delivers $8$ packets (to all users apart from the $4^{\text{th}}$), thus achieving a pudDoF of $d\left(\frac{5}{2}, 0\right)=\frac{8}{9}$.
}
\label{fig:mttwo}
\end{figure*}
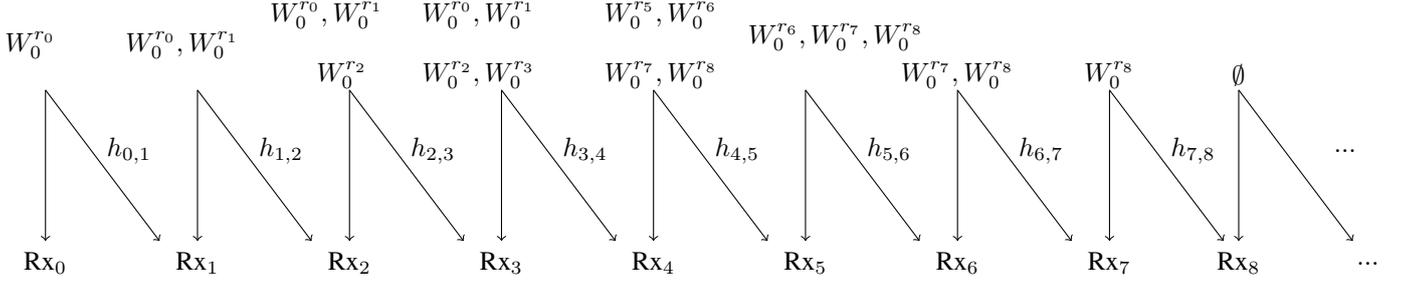
\subsection{Proof of Theorem~\ref{thm:nocaching}, Eq. \eqref{EqNonCache2}}

We explain in this section how the scheme presented in~\cite{ElGamal-Annapureddy-Veeravalli-ICC12} can be modified to prove the result in Theorem~\ref{thm:nocaching}, Eq.~\eqref{EqNonCache2}. The key idea is to employ time sharing for the scheme that achieves the same puDoF when the backhaul allows for distributing a message to a maximum of $x$ transmitters.

Similar to the above proof, the proposed scheme completes in $2x+1$ communication blocks, and each file is subpacketized into $2x$ subfiles, and each subfile is delivered through a one DoF link. In each communication block, the network is split into subnetworks, where each has $2x+1$ consecutive transmitter-receiver pairs.

Consider the case when $x=2$. In the proposed scheme, subfile $W_{0}^{r_{0}}$ is communicated between the first transmitter-receiver pair with no interference. The same subfile is also downloaded by the second transmitter (with index $1$) to cancel its interference at receiver $1$. Subfile $W_{0}^{r_{1}}$ is then delivered through transmitter $1$ to receiver $1$. Similarly, transmitter $3$ delivers $W_{0}^{r_{4}}$ to the last receiver in the first subnetwork, and transmitter $2$ downloads the same subfile to cancel its interference at receiver $3$. Finally, transmitter $2$ delivers $W_{0}^{r_{3}}$ to receiver $3$ with no interference. Note that $W_{0}^{r_{3}}$ is not transmitted in the first block of communication. Also, the last transmitter in the subnetwork is inactive to eliminate inter-subnetwork interference.

In each communication block, a total of $x(x+1)$ subfiles are downloaded from the backhaul for each subnetwork of $2x+1$ users. Upon the conclusion of all $2x+1$ communication blocks, each transmitter has downloaded an equal number of subfiles from the backhaul. Since each file has $2x$ subfiles, the per-transmitter backhaul load $M_T$ is given by,
\begin{equation}
    M_T=\frac{x(x+1)}{2x}=\frac{x+1}{2}.
\end{equation}

\section{Memory Sharing}\label{SecMemorySharing}
In this appendix, we describe the memory sharing concept, which we first use to prove the result of Th.~\ref{thm:caching}, Eq.~\eqref{EqCaching2}, i.e, the required backhaul load under the assumption of complete interference mitigation and a fractional cache size $\gamma=\frac{1}{2k}$, and we further use this result to calculate the puDoF of any pair $\left( M_{T}\in\mathbb{N}, \gamma<\frac{1}{4M_{T}}  \right)$.

The main idea of memory sharing is to split each file into two parts and cache from each part in an uneven manner. We begin by splitting each file $W^{n}\in\mathcal{F}$ into parts i.e., $W^{n}\to \{W^{n,1},W^{n,2}\}$ with respective sizes $|W^{n,1}|=p\cdot |W^{n}|$ and $|W^{n,2}|=(1-p)\cdot |W^{n}|$, where $p\in[0,1]$. We proceed to cache, at each user, a fraction $\gamma_{1}=\frac{1}{2x-1}$ from the first part of each file and a fraction $\gamma_{2}=\frac{1}{2x+1}$ from the second part of each file, which means that the cache constraint must satisfy
\begin{align}
    \gamma&=p\cdot\gamma_{1}+(1-p)\cdot\gamma_{2},\\
    \frac{1}{2x}&=p\frac{1}{2x-1}+(1-p)\frac{1}{2x+1}, \\ p&=\frac{2x-1}{4x}.
\end{align}

Then, using the result of Eq.~\eqref{EqCaching1}, for each of the two parts, we can calculate the total required backhaul load as
\begin{align}
    M_{T}&=p \frac{1-\gamma_{1}^{2}}{4\gamma_{1}} +(1-p)\frac{1-\gamma_{2}^{2}}{4\gamma_{2}}\\\nonumber
    &=\frac{2x-1}{4x}\frac{4(2x)(2x-2)}{2x-1}+\frac{2x+1}{4x}\frac{4(2x)(2x+2)}{2x+1}\ \ &\\\nonumber 
    &=8x=\frac{1}{4\gamma}.&\square
\end{align}

Further, in order to calculate the puDoF for an arbitrary $\gamma<\frac{1}{4M_{T}}$ that is paired with an integer-valued backhaul load, $M_{T}\in\mathbb{N}$, we follow the same procedure of splitting the file into two parts, where now the fractional cache sizes chosen for each part take the values $\gamma_{1}=\frac{1}{4M_{T}}$ and $\gamma_{2}=0$, respectively, thus $p$ can be computed by solving
\begin{align}
    \gamma=p\gamma_{1}+(1-p)\gamma_{2} \Rightarrow p=4 \gamma M_{T},
\end{align}
so that the memory cache constraint is respected. The puDoF when we transmit each part is going to be, respectively, $d(M_{T},\gamma_{1})=1$ (cf. Eq.~\eqref{EqCaching2}) and $d(M_{T},0)=\frac{4M_{T}-2}{4M_{T}-1}$ (cf. Eq.~\eqref{EqNonCache2}), thus the time required to serve all demands would be
\begin{align*}
	\mathcal{T}=p\frac{1-\gamma_{1}}{d(M_{T},\gamma_{1})}+(1-p)\frac{1-\gamma_{2}}{d(M_{T},0)},
\end{align*}
from which we can calculate the achievable per-user DoF as
\begin{align*}
	d\left(M_{T},\gamma\right)=\frac{1-\gamma}{\mathcal{T}}.
\end{align*}

\end{document}